\documentclass{article}
\usepackage{graphicx} % Required for inserting images
\usepackage{a4wide}
\usepackage{amsmath}
\usepackage{amsthm}
\usepackage{amsfonts}

\usepackage{amsmath,amsthm,amsfonts,amssymb}
\usepackage{thm-restate}
\usepackage{fullpage}
\usepackage[utf8]{inputenc}
\usepackage[dvipsnames]{xcolor}
\usepackage{xspace}
\usepackage[shortlabels]{enumitem}
\usepackage[hypertexnames=false,colorlinks=true,urlcolor=Blue,citecolor=Green,linkcolor=BrickRed]{hyperref}
\usepackage[capitalise]{cleveref}
\usepackage[OT4]{fontenc}
\usepackage[ruled]{algorithm}
\usepackage{ifpdf}
\usepackage{todonotes}
\usepackage{thmtools}
\usepackage{mathtools}
\usepackage[noend]{algpseudocode}
\usepackage{textpos}
\usepackage{makecell}
\usepackage{multirow}

\usepackage[margin=1in]{geometry}

\usepackage{crossreftools}
\theoremstyle{plain}
\newtheorem{theorem}{Theorem}[section]
\newtheorem{lemma}[theorem]{Lemma}

\newtheorem{fact}[theorem]{Fact}
\newtheorem{observation}[theorem]{Observation}

\algnewcommand{\IIf}[1]{\State\algorithmicif\ #1\ \algorithmicthen}
\algnewcommand{\EndIIf}{\unskip\ \algorithmicend\ \algorithmicif}

\newcommand{\dist}{\mathsf{dist}}
\newcommand{\Oh}{\mathcal{O}}
\newcommand{\N}{\mathbb{N}}
\newcommand{\tw}{\mathrm{tw}}

\newcommand{\Aa}{\mathcal{A}}
\newcommand{\Bb}{\mathcal{B}}
\newcommand{\Tt}{\mathcal{T}}
\newcommand{\Ll}{\mathcal{L}}
\newcommand{\mx}{\mathrm{max}}
\newcommand{\mn}{\mathrm{min}}
\newcommand{\bag}{\mathsf{bag}}

\newcommand{\bnd}{\partial}
\newcommand{\edges}{\mathsf{edges}}

\renewcommand{\leq}{\leqslant}
\renewcommand{\geq}{\geqslant}
\renewcommand{\le}{\leqslant}

\newcommand{\aff}[1]{\small\textcolor{black!60}{#1}}

\newcommand{\tsc}[1]{\textsc{\text{#1}}}

\def\DEBUG{true}
\ifdefined\DEBUG{}

\newcommand{\ms}[1]{\textcolor{violet}{\textbf{Marek: } #1}}
\def\rem#1{{\marginpar{\raggedright\scriptsize #1}}}
\newcommand{\mwr}[1]{\rem{\textcolor{blue}{MW: #1}}}

\newcommand{\msr}[1]{\rem{\textcolor{violet}{MS: #1}}}
\newcommand{\micr}[1]{\rem{\textcolor{pink}{MP: #1}}}
\else

\newcommand{\ms}[1]{ }
\newcommand{\mwr}[1]{ }
\newcommand{\msr}[1]{ }
\newcommand{\micr}[1]{ }
\fi

\title{\Huge{Dynamic Detours}\thanks{This research has been initiated/conducted at the AlgUW workshop (Będlewo 09.2025), supported by the Excellence Initiative -- Research University (IDUB) funds of the University of Warsaw. The work of MiP on this manuscript is a part of project BOBR that has received funding from the
European Research Council (ERC) under the European Union’s Horizon 2020 research and innovation programme
(grant agreement No. 948057).
AR was supported by Polish National Science Centre SONATA BIS-12 grant number 2022/46/E/ST6/00143.
MW was supported by Polish National Science Centre SONATA-19 grant number 2023/51/D/ST6/00155.} }

\author{
	Daniel Dadush \\
	\aff{CWI and Utrecht University} \\
    \aff{The Netherlands} \\
	\aff{dadush@cwi.nl}
    \and
	Micha\l{} Pilipczuk \\
	\aff{University of Warsaw} \\
	\aff{Warsaw, Poland}\\
    \aff{michal.pilipczuk@mimuw.edu.pl}\vspace{0.4cm}
    \and
	Amadeus Reinald \\
	\aff{University of Warsaw} \\
	\aff{Warsaw, Poland}\\
	\aff{reinald@mimuw.edu.pl}
    \and
	  Marek Soko\l{}owski \\
	\aff{Max Planck Institute for Informatics} \\
    \aff{Saarland Informatics Campus} \\
	\aff{Saarbrücken, Germany}\\
	\aff{msokolow@mpi-inf.mpg.de}
    \and
	Micha\l{} W\l{}odarczyk \\
	\aff{University of Warsaw} \\
	\aff{Warsaw, Poland}\\
	\aff{michal.wloda@gmail.com}
}
\date{}

\begin{document}

\maketitle

 \begin{textblock}{20}(-1.75, 6.1)
 \includegraphics[width=40px]{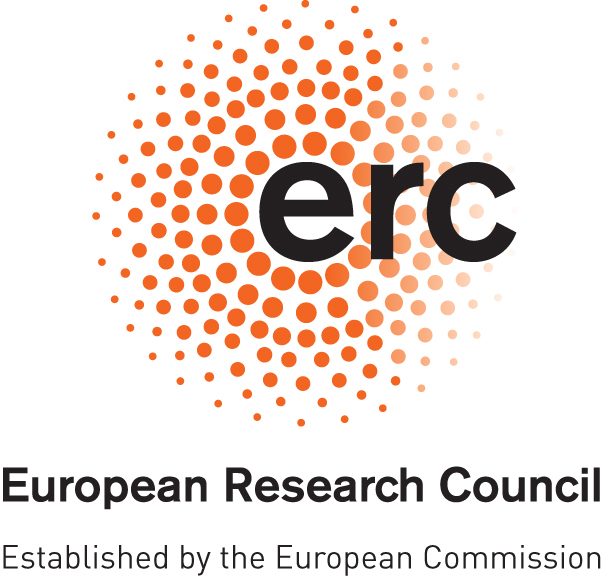}%
 \end{textblock}
 \begin{textblock}{20}(-1.75, 7.1)
 \includegraphics[width=40px]{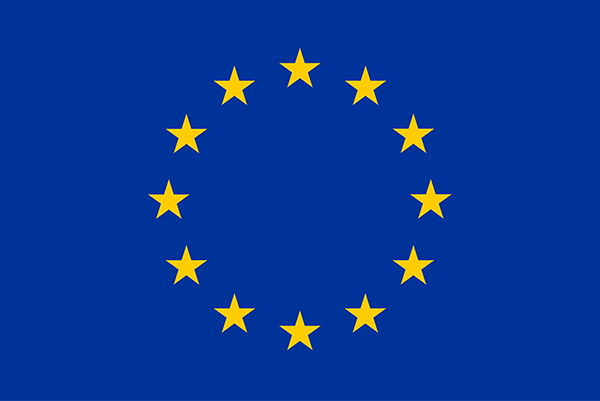}%
 \end{textblock}
\begin{abstract}
	Fix a parameter $k\in \N$. We give dynamic data structures that for a fully dynamic undirected graph~$G$, updated over time by edge insertions and edge deletions, can answer the following queries:
	\begin{itemize}
		\item Long $(u,v)$-path: Given $u,v\in V(G)$, is there a path from $u$ to $v$ of length at least $k$?
		\item Long $(u,v)$-detour: Given $u,v\in V(G)$, is there a path from $u$ to $v$ of length at least $\dist_G(u,v)+k$?
		\item Even/odd $(u,v)$-path: Given $u,v\in V(G)$, is there a path from $u$ to $v$ of even/odd length?
	\end{itemize}
    \noindent
	The amortized time of executing an update or answering a query is $2^{\Oh(k^3)} \log n + \Oh(\log^2 n \log^2 \log n)$ in the first two cases, and $\Oh(\log^2 n \log^2 \log n)$ in the last, where $n$ is the number of vertices of $G$.

    The first result is in sharp contrast with
    known conditional lower bounds for reporting paths of length at most $k$.
    Specifically, there is no data structure supporting queries about $(u,v)$-paths of length at most two in time $n^{o(1)}$ unless the Triangle Conjecture fails.

    Our main technical contribution is a mechanism of ``delayed edge insertion'' that works locally on the level of biconnected components.
\end{abstract}

\section{Introduction}

The area of \emph{parameterized dynamic data structures} aims to apply the principles of parameterized complexity to the setting of data structures for dynamically changing inputs. That is, instead of designing a static algorithm for a parameterized problem that just reads the input and computes the answer, the goal is to propose a data structure that efficiently maintains the answer to the problem under updates to the maintained instance. Following the parameterized paradigm, the complexity guarantees for updates and queries to the data structure can now be measured both in terms of the considered parameters --- where we typically allow superpolynomial dependence --- and in terms of the total instance size --- where we wish to keep the dependence sublinear, or even (poly)logarithmic or constant. 

It is natural to ask which classic results and techniques of parameterized complexity can be lifted to the setting of dynamic data structures. So far, a number of positive results have been found, including data structures for detection and counting of small subgraph patterns~\cite{AlmanMW20,ChenCDFHNPPSWZ21,DvorakT13}, parameterized graph modification problems~\cite{IwataO14,MajewskiPZ24}, kernelization~\cite{IwataO14,BannachHRT22,BertramHJK25}, problems in topologically-constrained and geometric graphs~\cite{AnCJJL0SSS24,KorhonenNPS24}, parameterized string problems~\cite{OlkowskiPRWZ23}, or even problems related to automata~\cite{GrezMPPR22}. Of particular importance is the recent direction of dynamic maintenance of graph decompositions related to classic width parameters, such as treedepth~\cite{ChenCDFHNPPSWZ21,DvorakKT14}, feedback vertex number~\cite{MajewskiPS23}, treewidth~\cite{dynamicTWKorhonenETAL,Korhonen25}, or cliquewidth~\cite{KorhonenS24}. Indeed, such data structures typically also allow maintaining suitable dynamic programming tables for any problem that can be solved using the relevant decomposition by means of a bottom-up dynamic programming. Thus, dynamic data structures for maintaining graph decompositions typically lead to results that apply not to single problems, but to whole classes of problems. This is probably best exemplified by the dynamic counterparts of Courcelle's Theorem for treewidth~\cite{dynamicTWKorhonenETAL,Korhonen25} and for cliquewidth~\cite{KorhonenS24}.

In this work, we add another point to this growing list of results: we give parameterized dynamic data structures for finding long paths with prescribed endpoints. Our first result is the following.

\begin{restatable}{theorem}{thmpath}\label{thm:k-path-dynamic-DS}
	For every $k\in \N$ there exists a data structure that for a fully dynamic $n$-vertex graph $G$, supports the following operations in amortized $2^{\Oh(k^3)} \log n + \Oh(\log^2 n \log^2 \log n)$ time.
	\begin{itemize}[nosep]
		\item \tsc{Insert}($u,v$): insert an edge $uv$, provided it does not exist in $G$;
		\item \tsc{Delete}($u,v$): delete an edge $uv$, provided it exists in $G$;
		\item \tsc{LongPath}$(u,v)$: decide whether there is a path of length at least $k$ between $u$ and $v$.
	\end{itemize}
	The data structure can be initialized for an edgeless $G$ in time $2^{\Oh(k^3)}\cdot n$. 
\end{restatable}

%\mw{
This should be compared to the task of reporting $(u,v)$-paths of length {\em at most $k$}.
An easy reduction shows that even a static data structure reporting if $\dist_G(u,v) \le k$ in time $g(k,n)$ can be used to detect a $(k+1)$-cycle in time $m\cdot g(k,n)\cdot 2^{\Oh(k \log k)}\log n$, %\mwr{I expanded this expression because I guess the reduction needs color coding}
where $m$ denotes the number of edges~\cite{AlmanMW20, GutenbergWW20}.
Therefore, such a data structure with query time $f(k)\cdot n^{o(1)}$ would break the Triangle Conjecture, which asserts that triangle detection requires time $m^{1+\Omega(1)}$.
What is more, it would also break the 3SUM Conjecture~\cite{AlmanMW20}.
Next, the $k$-Cycle Hypothesis states that for every $\varepsilon > 0$ there exists $k$ so that detecting a cycle of length $k$ requires time $\Omega_k(m^{2-\varepsilon})$~\cite{AbboudW14, GutenbergWW20}.
Under this assumption, reporting if $\dist_G(u,v) \le k$ cannot be even done in time $\Oh_k(m^{1-\varepsilon})$ per query, for any universal constant $\varepsilon > 0$.

We remark that the $\log n$ dependence on the number of vertices in the runtime is necessary: for $k = 1$, the data structure solves precisely the problem of dynamic connectivity, for which a~classic $\Omega(\log n)$ lower bound was given by P{\u{a}}tra{\c{s}}cu and Demaine~\cite{PatrascuD04}.

Finding long paths was among the first parameterized problems considered from the point of view of dynamic data structures.
Alman, Mnich, and Vassilevska-Williams~\cite{AlmanMW20} used a dynamic variant of the color-coding technique to give a fully dynamic data structure with update time $2^{\Oh(k\log k)}\cdot \log n$ that can report whether the maintained graph contains a path of length $k$.
%\mw{
This data structure can be easily adjusted to 
report a $(u,v)$-path of length {\em at least $k$} 
or conclude there is no $(u,v)$-path of length {\em exactly $k$}.
However, if the shortest $(u,v)$-path has length $k-1$ and the second shortest one has length $\Omega(n)$, color-coding alone seems not sufficient to obtain the guarantees of \Cref{thm:k-path-dynamic-DS}.%}

Later, Chen et al.~\cite{ChenCDFHNPPSWZ21} gave a dynamic data structure for detecting a path of length $k$ with amortized update time $2^{\Oh(k^2)}$; thus, independent of $n$. Their approach is based on a dynamic data structure for maintaining an optimum treedepth decomposition of the graph, and crucially relies on the fact that any graph of large treedepth must contain a long path. Therefore, the approach inherently cannot handle queries about paths with prescribed roots.

Our technique allows us to not only detect $(u,v)$-paths of length at least $k$, but also $(u,v)$-paths that are longer by $k$ than the distance between $u$ and $v$. Such paths are often called {\em{detours}}~\cite{bezakova2019findingDetoursFPT, akmal2023, FominGLSSS23, JacobWZ24, BergerSS21}.

\begin{restatable}{theorem}{thmDetour}\label{thm:detour-dynamic-DS}
	For every $k\in \N$ there exists a data structure that for a fully dynamic $n$-vertex graph $G$, supports the following operations in amortized $2^{\Oh(k^3)} \log n + \Oh(\log^2 n \log^2 \log n)$ time.
	\begin{itemize}[nosep]
		\item \tsc{Insert}($u,v$): insert an edge $uv$, provided it does not exist in $G$;
		\item \tsc{Delete}($u,v$): delete an edge $uv$, provided it exists in $G$;
		\item \tsc{LongDetour}$(u,v)$: decide whether there is a path of length at least $\dist_G(u,v)+k$ between $u$ and~$v$.
	\end{itemize}
	The data structure can be initialized for an edgeless $G$ in time $2^{\Oh(k^3)}\cdot n$. 
\end{restatable}

The static problem underlying \cref{thm:detour-dynamic-DS} is called {\sc{Long Detour}}: given an undirected graph $G$, vertices $u,v$, and parameter $k$, decide whether there is a simple path from $u$ to $v$ of length at least $\dist_G(u,v)+k$ in~$G$. An FPT algorithm for this problem was first given by Bez\'akov\'a, Curticapean, Dell, and Fomin~\cite{bezakova2019findingDetoursFPT}. In fact, as we will later discuss, the approach proposed by Bez\'akov\'a et al.\ in~\cite{bezakova2019findingDetoursFPT} is the main point of inspiration for our proofs of \cref{thm:k-path-dynamic-DS,thm:detour-dynamic-DS}, along with the results of Korhonen et al.~\cite{dynamicTWKorhonenETAL} and of Korhonen~\cite{Korhonen25} for the {\em{dynamic treewidth problem}} --- maintaining an approximate tree decomposition of a fully dynamic graph of bounded treewidth.

Finally, we show the robustness of our approach by applying it to a non-parameterized problem: detection of a path of given parity between a prescribed pair of vertices. Specifically, we prove the following.

\begin{theorem}\label{thm:parity-DS}
	There exists a data structure that for a fully dynamic $n$-vertex graph $G$, supports the following operations in amortized $\Oh(\log^2 n \log^2 \log n)$ time.
	\begin{itemize}[nosep]
		\item \tsc{Insert}($u,v$): insert an edge $uv$, provided it does not exist in $G$;
		\item \tsc{Delete}($u,v$): delete an edge $uv$, provided it exists in $G$;
		\item \tsc{EvenPath}$(u,v)$: decide whether there is a path of even length between $u$ and $v$;
		\item \tsc{OddPath}$(u,v)$: decide whether there is a path of odd length between $u$ and $v$.
	\end{itemize}
	The data structure can be initialized for an edgeless $G$ in time $\Oh(n)$. 
\end{theorem}

\paragraph*{Techniques.}

The main engine behind all our theorems is a stronger variant of the {\em delayed edge insertion} technique. 
Eppstein et al.~\cite{EppsteinGIS96} observed that dynamic planarity testing can be reduced to maintaining a dynamic planar graph and checking if a specified edge insertion violates planarity.
From the perspective of amortized operation time, one can simply keep a pile of rejected edges and try inserting them again after some edge is removed.
Naturally, this trick works for other graph properties as well~\cite{ChenCDFHNPPSWZ21, dynamicTWKorhonenETAL}.
It can also be implemented on the level of connected components by keeping a separate pile for each component~\cite{Holm20}.

For the sake of detecting a $(u,v)$-path of length at least $k$, we want to apply the win-win approach from~\cite{bezakova2019findingDetoursFPT, ChenCDFHNPPSWZ21}: either the graph has sufficiently large treewidth (or treedepth) and the answer is yes, or treewidth is small and we can employ dynamic programming.
The issue is that large treewidth does not ensure the existence of a long $(u,v)$-path for an arbitrary pair $(u,v)$.
For this, a stronger condition is necessary:
there needs to be a~biconnected component of large treewidth, located {\em between} $u$ and $v$.

Our strategy is to maintain a graph of bounded treewidth~\cite{Korhonen25} and for each biconnected component implicitly store a pile of rejected edges --- those that make the treewidth grow beyond some threshold.
This is significantly harder than maintaining such piles on the level of connected components because single edge insertion can merge $\Omega(n)$ biconnected components into one while single edge deletion can revert this transformation.
To mitigate this, we keep the rejected edges in a separate data structure: we adapt the dynamic biconnectivity data structure by Holm, Nadara, Rotenberg and Sokołowski~\cite{HolmNRS2025Stoc} to support the following operations: (1) mark edge $e$, and (2) return some marked edge (if there is any) in the biconnected component containing vertices $u,v$.
By combining these two data structures, we can effectively say that either there is a biconnected component of large treewidth between $u$ and $v$ (and then the answer is yes) or that this part of the graph has small treewidth, which allows us to read the answer from a DP state in the dynamic tree decomposition. 
%\ama{I think this is probably the very good way to quickly introduce the intuition of the proof, because this was the spirit. If I had to nitpick, I'm wondering if we want to slightly adjust it to stick more precisely to what we do later (as I have rewritten). Mostly, would it be more natural to already introduce the term "relevant part", which captures this "between" that is a bit vague here? Also, with the current writeup we could think that biconnected components between $s$ and $t$ are treated separately, but we will really only be using high tw of the relevant part/chain of components as a whole.}
%\mwr{I agree but I deliberately sweep a lot under the carpet to keep this part short. If you think some vocabulary should be adjusted to the later sections, please add a suggestion for replacement sentence}
%\ama{It's good like this, I was just nitpicking but I think it would be less intuitive to try to get closer to the techniques in the proof and introduce "relevant part".}

\section{Preliminaries}

In this work, all graphs are simple, finite and undirected.
If $S \subseteq V(G)$, then $G[S]$ denotes the subgraph of $G$ induced by $S$.
For $u, v \in V(G)$, we write $G + uv$ (respectively, $G - uv$) to describe the supergraph (subgraph) of $G$ formed by adding (removing) the edge $uv$.
Note that $G + uv = G$ if $uv \in E(G)$, and similarly $G - uv = G$ if $uv \notin E(G)$.
More generally, if $H$ is a~graph with $V(H) \subseteq V(G)$, then $G + H$ (resp., $G - H$) is formed from $G$ by adding all edges of $H$ that are absent from $G$ (resp., removing all edges of $H$ present in $G$).

% \subsection{definitions}

\subsection{Combinatorial properties of graphs}

\paragraph{Biconnectivity and relevant subgraphs.}
We present two definitions of \emph{biconnectivity} in graphs: one describing it as a~relation on the vertices of the graph, and the other phrasing it as an~equivalence relation on the edges of the graph.
Namely, we say that two vertices $u$, $v$ of a~graph $G$ are biconnected if either $uv \in E(G)$, or there exist two paths, vertex-disjoint except from their endpoints, connecting $u$ and $v$.
Equivalently, $u$ and $v$ are in the same biconnected component of $G$ if $u$ and $v$ cannot be separated by removing %a~single edge (a~\emph{bridge}) from $G$ or
a~single vertex (a~\emph{cut-vertex}) different than $u$ and $v$.
Next, we will say that two edges $e$, $f$ are biconnected if there exists a~simple cycle in $G$ including both $e$ and $f$ in its edge set.
It is standard that the biconnectivity relation is an~equivalence relation on the edges of $G$ (but not the vertices of $G$)~\cite{Harary1969,TarjanV1985}.

We borrow the following definition of a~\emph{relevant subgraph} from~\cite{bezakova2019findingDetoursFPT}.
Let $G$ be a graph, and let $s, t \in  V (G)$. The \emph{$(s, t)$-relevant part} of $G$ is the graph $G_{s,t}$ induced by all vertices that are contained in at least one $(s, t)$-path.
The following lemma shows that we can recover $G_{s,t}$ dynamically, provided we can maintain biconnected components, motivating our use of dynamic biconnectivity data structures.
\begin{observation}\label{obs:Gst-equals-block-B}
    For any graph $G$, and any two vertices $s,t \in V(G)$, $s$ and $t$ are in the same biconnected component $B$ of $G + st$, and $B = G_{s,t} + st$.
\end{observation}
\begin{proof}
    Indeed, if $s,t$ are not joined by a path of length at least $2$ in $G$ (in particular if they are not connected), $B$ consists only of $st$, and satisfies the observation.
    Otherwise, any $(s,t)$-path $P$ along with $st$ forms a cycle in $G+st$, meaning $V(G_{s,t}) \subseteq V(B)$.
    Conversely, for any $b \in V(B)$, Menger's theorem provides us with a path from $b$ to $s$ and one from $b$ to $t$ that are disjoint, and concatenating them yields $b \in V(G_{s,t})$.
    That is, $V(B) = V(G_{s,t})$, and since $st$ is the only edge of $B$ that may not be present in $G_{s,t}$, $B = G_{s,t} + st$. 
\end{proof}

\paragraph{Treewidth.}
We use the classical notion of treewidth \cite{RobertsonS84}.
A~\emph{tree decomposition} of a~graph $G$ is a~pair $\mathcal{T} = (T, \bag)$, where $T$ is a~tree and $\bag \,\colon\, V(T) \to 2^{V(G)}$ (a mapping of nodes of $T$ to sets of vertices, called \emph{bags}), with the properties that: (i) each vertex of $G$ belongs to a~non-empty set of bags inducing a~connected subtree of $T$, and (ii) for each edge of $G$, both of its endpoints belong together to some bag of $\mathcal{T}$.
The \emph{width} of the tree decomposition is the maximum cardinality of any bag, minus $1$.
The \emph{treewidth} of $G$, denoted $\tw(G)$, is the minimum possible width of a~tree decomposition of $G$.

In this work we implicitly use several combinatorial properties of treewidth: for $u, v \in V(G)$, we have $\tw(G) \leq \tw(G + uv) \leq \tw(G) + 1$; and $\tw(G)$ is equal to the maximum treewidth of the biconnected components of $G$ (this follows e.g.\ from~\cite[Lemma 1]{DemaineHT02}).

\subsection{Dynamic data structures}

\paragraph{Dynamic biconnectivity.}
Our work heavily relies on data structures that efficiently maintain the structure of biconnected components in a~fully dynamic graph.
Here, we use the data structure by Holm, Nadara, Rotenberg and Sokołowski~\cite{HolmNRS2025Stoc}:

\begin{theorem}[{\cite[Theorem 1]{HolmNRS2025Stoc}}]
    \label{thm:biconnectivity-ds-basic}
    There exists a~data structure maintaining a~fully dynamic $n$-vertex undirected graph $G$ in the word RAM model with $\Omega(\log n)$ word size. %\mwr{Can we state this assumption at the beginning so that we dont have to restate it in lemmas?}
    The data structure can handle the following updates and queries for vertices $u,v \in V(G)$:
    \begin{itemize}[nosep]
        \item \tsc{Insert}($u,v$) and \tsc{Delete}($u,v$): insert or remove the edge $uv$;
        \item \tsc{IsBiconnected}($u,v$): determine whether two vertices $u$ and $v$ are biconnected.
    \end{itemize}
    Each operation can be performed in amortized time $\Oh(\log^2 n \log^2 \log n)$.
\end{theorem}

For the purposes of our work, we need to extend the data structure above to support \emph{marking} edges.
Namely, some edges of the dynamic graph can be designated as \emph{marked}, and we want to additionally support searching marked edges in the biconnected components of the graph.
Formally, we show the following:
\begin{theorem}
    \label{thm:biconnectivity-ds}
    The data structure from \Cref{thm:biconnectivity-ds-basic} can be extended to support \emph{marking} some of the edges of the graph.
    In this regime, the data structure can also support the following operations for an edge $e=uv$ of the current graph $G$:
    \begin{itemize}[nosep]
        \item \tsc{Mark}($uv$) and \tsc{Unmark}($uv$): mark or unmark $uv$;
        \item \tsc{FindMarkedEdge}($uv$): return a marked edge $f$ biconnected with $uv$, or correctly output \textsf{Null} if no such edge exists.
    \end{itemize}
    The amortized time complexity remains at $\Oh(\log^2 n \log^2 \log n)$.
\end{theorem}

The proof of \Cref{thm:biconnectivity-ds} follows by inspecting the implementation of the original biconnectivity data structure: there, biconnectivity in a~dynamic graph $G$ is maintained using an~intricate dynamic \emph{tree} data structure supporting an~extensive list of updates and queries, including marking \emph{vertices} of the tree and finding marked \emph{vertices} in parts of the tree corresponding to the biconnected components of~$G$.
By suitably extending this tree data structure, we show that it can also be exploited to locate marked \emph{edges} in the biconnected components of $G$.
We move the formal exposition to \Cref{sec:biconnectivity-ds-marking}.

\paragraph{Dynamic treewidth.}

We also rely on the dynamic treewidth data structure of Korhonen~\cite{Korhonen25}. Specifically, we will need to maintain a tree-decomposition of some width $t$, as well as the ability to query for shortest and longest paths in $G$.
The following lemma formalizes this.
Its proof is a standard construction of a tree decomposition automaton, which we defer to~\Cref{sec:tree-automata-path-query}.
\begin{lemma}
    \label{lem:dynamic-treewidth}
    Let $t \in \mathbb{N}$.
    There exists a~data structure maintaining an~$n$-vertex undirected graph $G$ and supporting the following operations for two vertices $u,v$:
    \begin{itemize}[nosep]
        \item \tsc{Insert}($u,v$): insert an edge with endpoints $u$, $v$
        as long as $\tw(G + uv) \leq t$; if  $\tw(G + uv) > t$, the insertion is aborted,
        %. The operation always succeeds if $\tw(G + uv) \leq t$; otherwise, the data structure may abort the insertion;
        %in the latter case the insertion is aborted,
        \item \tsc{Delete}($u,v$): remove edge with endpoints $u$, $v$,
        \item \tsc{ShortestPath}($u,v$) and \tsc{LongestPath}($u,v$): output the lengths of the shortest and longest $(u,v)$-paths in $G$; or $+ \infty$ and $-\infty$, respectively, if there is no $(u,v)$-path.
    \end{itemize}
    Each operation can be performed in amortized time $2^{\Oh(t^3)} \cdot \log n$.
\end{lemma}

\paragraph{Dynamic bipartiteness.}

In the proof of \Cref{thm:parity-DS} we will need a data structure that maintains a bipartite graph and reports when inserting some edge would create an odd cycle.
It is well-known that dynamic bipartiteness can be implemented in the same way as dynamic pairwise connectivity~\cite{HenzingerK99}.
We give a simple black-box reduction and employ the deterministic dynamic connectivity data structure~\cite{HolmLT01}.
Note that similar data structures have already been proposed (see e.g., \cite[Theorem 8]{KashyopNNP23}), but for completeness we give an~interface and an~implementation of such a~data structure below.
% \ms{Chatty suggests it's been done here: \cite[Theorem 8]{KashyopNNP23}, I tend to agree}
% \mwr{Then we can add a pointer to them and move the proof to the appendix}

\begin{lemma}
    \label{lem:dynamic-bipartite}
    There exists a~data structure maintaining an~$n$-vertex undirected bipartite graph $G$ and supporting the following operations:
    \begin{itemize}[nosep]
        \item \tsc{Insert}($u,v$): insert edge $uv$ as long as the insertion does not create an odd cycle; if $G + uv$ has an odd cycle, the insertion is aborted,
        \item \tsc{Delete}($u,v$): remove edge $uv$ from the graph,
        \item \tsc{EvenPath}($u,v$) and \tsc{OddPath}($u,v$): decide whether there is a~$(u,v)$-path of even (odd) length.
    \end{itemize}
    Each operation can be performed in amortized time $\Oh(\log^2 n)$.
\end{lemma}
\begin{proof}
Consider graph $\widehat G$ obtained from $G$ by replacing each vertex $v$ with two copies $v_0,v_1$ and replacing each edge $uv$ with two edges $v_0u_1$, $v_1u_0$.
We rely on the following easy observation.

\begin{fact}\label{fact:dynamic-bipartite}
Graph $G$ contains an odd cycle passing through a vertex $v$ if and only if graph $\widehat G$ contains a path from $v_0$ to $v_1$.
\end{fact}

Our data structure maintains the graph $\widehat G$ using the connectivity data structure from~\cite{HolmLT01}, that is, \tsc{Insert}($u,v$) is implemented by insertion of $v_0u_1$, $v_1u_0$, and likewise for \tsc{Delete}($u,v$).
When an edge $uv$ is being added, we need to report if $G + uv$ is still bipartite.
Since insertions violating bipartiteness are rejected, we work under the assumption that $G$ is bipartite.
Observe that a hypothetical odd cycle in  $G + uv$ must use the edge $uv$, hence it must pass through $v$.
By \Cref{fact:dynamic-bipartite}, this can be detected by testing if $v_0,v_1$ are connected in $\widehat G + v_0u_1 + v_1u_0$.
If they are, we revert the two edge insertions in $\widehat G$ and report that $G + uv$ has an odd cycle.
Also, observe that \tsc{EvenPath}($u,v$) (respectively, \tsc{OddPath}($u,v$)) can be implemented by verifying the connectivity in $\widehat{G}$ between $u_0$ and $v_0$ (resp., $u_0$ and $v_1$).

Each operation in our data structure invokes $\Oh(1)$ operations in the connectivity data structure for $\widehat G$.
Since the data structure from~\cite{HolmLT01} handles each operation in amortized time  $\Oh(\log^2 n)$, the same applies to our case.
\end{proof}

\section{Dynamic $(s,t)$\tsc{-Path} and $(s,t)$\tsc{-Detour}}
\label{sec:dynamic-long-paths}

In this section we prove \Cref{thm:k-path-dynamic-DS,thm:detour-dynamic-DS} by proposing appropriate dynamic data structures testing for long $(s,t)$-paths and $(s,t)$-detours.
We begin by setting up the combinatorial foundations of our data structures in \Cref{ssec:paths-from-high-tw}, and then we implement the data structures themselves in~\Cref{ssec:dynamic-path,ssec:dynamic-detour}.

\subsection{Extracting long paths and detours from high-treewidth graphs}
\label{ssec:paths-from-high-tw}

Before defining the data structures, we show here a few structural results ensuring that, for a~given pair $s, t$ of vertices of $G$, if the $(s,t)$-relevant part of $G$ has sufficiently large treewidth, then $G$ admits both long $(s,t)$-paths and long $(s,t)$-detours.

For $(s,t)$-\tsc{Path}, we will need the following folklore result to ensure the existence of a $k$-path.
\begin{lemma}\label{lem:extract-k-path-high-tw}
    If $B$ is a biconnected graph such that $\tw(B) \geq 2k-1$, then for any $s,t \in V(B)$, there exists an $(s,t)$-path of length at least $k$. Moreover, for any graph $G$ and $s,t \in V(G)$ such that $\tw(G_{s,t}) \geq 2k-1$, the same holds.
\end{lemma}
\begin{proof}
      Since $B$ has treewidth at least $2k-1$, a result of Birmelé~\cite{birmeleCircumference} ensures that $B$ contains a cycle $C$ of length at least $2k$.
      If $k=1$, the result holds trivially. If $k \geq 2$, $B$ cannot be a single edge, and for any $s,t \in V(B)$, Menger's theorem applied to $\{s,t\}$ and $C$ yields a path from $s$ to $c_s \in C$ and a vertex-disjoint path from $t$ to $c_t \in C$. Then, taking these two paths along with longer arc of $C$ between $c_s$ and $c_t$ yields a path of length at least $|C| / 2 \geq k$ between them.
      Finally, note that for any $G$ and $s,t \in V(G)$, if $\tw(G_{s,t}) \geq 2k-1$, then \Cref{obs:Gst-equals-block-B} ensures the biconnected component of $G+st$ containing $s,t$ is exactly $G_{s,t} + st$, so $\tw(G_{s,t} + st) \geq 2k-1$ and applying the first part of the lemma to $G_{s,t} + st$ concludes.
\end{proof}

The analogue of~\Cref{lem:extract-k-path-high-tw} for $(s,t)$-\tsc{Detour} was shown in~\cite{bezakova2019findingDetoursFPT}: imposing higher treewidth (in the relevant part of $G$) guarantees also a long detour.
\begin{theorem}[Theorem 3.6 in~\cite{bezakova2019findingDetoursFPT}]\label{thm:extract-k-detour-high-tw}
    For any graph $G$ and $s,t \in V(G)$ such that $\tw(G_{s,t}) > 32k+46$, there exists an $(s,t)$-path of length at least $\dist_G(s,t)+k$.
\end{theorem}

\subsection{Dynamic $(s,t)$-Path}
\label{ssec:dynamic-path}

\paragraph{Dynamic data structure.}
For any fixed $k$, our data structure for $(\geq k)$-paths instantiates the following auxiliary data structures for a fully dynamic graph $G$:
\begin{itemize}[nosep]
    \item A biconnectivity data structure \tsc{bc} for $G$ given by~\Cref{thm:biconnectivity-ds},
    \item A dynamic treewidth data structure \tsc{tw} of a \emph{subgraph} $H \subseteq G$ such that $\tw(H) \leq 2k$ given by~\Cref{lem:dynamic-treewidth}.
\end{itemize}
Then, our calls to the methods of these data structures are prefixed with \tsc{bc} and \tsc{tw}, respectively.
The subgraph $H$ contains all the edges of $G$ whose \emph{last} insertion was accepted by \tsc{tw}.
Then, the marked edges of (biconnected components of) $G$ are exactly $E(G) \setminus E(H)$.
We stress that if an~edge $e \in E(G)$ is not in $H$, we only know that every past attempt to insert $e$ into \tsc{tw} has failed (even though the insertion of $e$ into $H$ may be currently possible).
This phenomenon also naturally appears in Eppstein's technique of delaying invariant-breaking insertions~\cite{EppsteinGIS96}.

The implementation of our data structure is given in \Cref{alg:insert,alg:delete,alg:cap}.
Insertion of an~edge $uv$ (\Cref{alg:insert}) resolves simply to attempting the insertion of $uv$ to \tsc{bc} (which always succeeds) and \tsc{tw} (which may fail); if \tsc{tw} rejects the insertion, we mark $uv$ in \tsc{bc} so as to preserve the invariant.
Removing the edge $uv$ (\Cref{alg:delete}) is also simple: we simply drop $uv$ from the auxiliary data structures.
Finally, in order to test if $G$ contains an~$(s,t)$-path of length at least $k$ (\Cref{alg:cap}), we begin by temporarily adding a~helper edge $st$ to \tsc{bc}, if not already present in $G$.
We then iterate the marked edges in the biconnected component containing $st$, successively adding them to \tsc{tw} until either all such marked edges are exhausted, or some insertion is aborted due to the treewidth of the graph becoming too large.
In the former case, all edges of $G_{s,t}$ are present in \tsc{tw} and so the length of the longest $(s,t)$-path can be simply queried in \tsc{tw}; in the latter, we claim that a~sufficiently long $(s,t)$-path exists due to \Cref{lem:extract-k-path-high-tw}.

\begin{algorithm}
    \caption{Insert}\label{alg:insert}
    \begin{algorithmic}
        \Function{\tsc{Insert}}{$u,v$}
        \State \tsc{bc.Insert}($u,v$)
        \If{\tsc{tw.Insert}($u,v$) aborts}
            \State \tsc{bc.Mark}($uv$)
        \EndIf
        \EndFunction
    \end{algorithmic}
\end{algorithm}

\begin{algorithm}
    \caption{Delete}\label{alg:delete}
    \begin{algorithmic}
        \Function{\tsc{Delete}}{$uv$}
        \State \tsc{bc.Delete}($uv$)
        \State \tsc{tw.Delete}($uv$)
        \EndFunction
    \end{algorithmic}
\end{algorithm}

\newcommand{\HelperEdgeVar}{\textsf{hasHelperEdge}\xspace}

\begin{algorithm}
\caption{LongPath}\label{alg:cap}
\begin{algorithmic}
\Function{\tsc{LongPath}}{$s,t$}
\State \HelperEdgeVar $\gets$ \textsf{False}
\If{$st \notin E(G)$}
    \State \tsc{bc.Insert}$(s,t)$
    \State \HelperEdgeVar $\gets$ \textsf{True}
\EndIf
%\If{\textbf{not} \tsc{bc.IsBiconnected}$(s,t)$}
%    \IIf{\textbf{not} \HelperEdgeVar} \tsc{bc.Delete}$(st)$
%   \State \Return \textbf{No}
%\EndIf
\While{\textsf{True}}
    \State $e \gets$ \tsc{bc.FindMarkedEdge}$(st)$
    %\State \ms{Something happens if no marked edge exists?}
    \IIf{$e = \sf{Null}$} \textbf{break}
    \If{\tsc{tw.Insert}$(e)$ aborts} \Comment{$\tw(H + e) > 2k$}
        %\State \tsc{tw.Delete}$(H,e)$
        \IIf{\HelperEdgeVar} \tsc{bc.Delete}$(st)$
        \State \Return \textbf{Yes}
    \Else
        \State \tsc{bc.Unmark}($e$)
    \EndIf
\EndWhile
\IIf{\HelperEdgeVar} \tsc{bc.Delete}$(st)$
\State \Return \tsc{tw.LongestPath}$(s,t)$ $\geq k$
\EndFunction
\end{algorithmic}
\end{algorithm}

The following lemma proves the correctness of \tsc{LongPath}.

%\ms{Comment: I \textbf{really} prefer to define $B$ as the biconnected component containing $st$ only locally in the setting of this lemma. Not in \Cref{lem:extract-k-path-high-tw} since it's a general statement about biconnected graphs.} 
\begin{lemma}\label{lem:k-path-query-soundness}
    For any $s,t \in V(G)$, \tsc{LongPath}$(s,t)$ correctly outputs whether $s$ and $t$ are joined by a path of length at least $k$.
\end{lemma}
\begin{proof}
    We fix $G$ to be the initial graph before the query, and $H \subseteq G$ the graph held by the treewidth data structure. Then, let $G',H'$ be the final states of the graphs $G,H$, respectively, at the end of the run of $\tsc{LongPath}(s, t)$, just before the helper edge $st$ is (potentially) removed from \tsc{bc}.
    %First, if $s$ and $t$ are not connected in $G$, they cannot belong to the same biconnected component of $G'$ after the addition of $st$, meaning \tsc{LongPath} correctly outputs \textbf{No}.
    %Assume now that $s,t$ are connected in $G$.

    Then, the first steps of the algorithm ensure $st \in E(G')$ in any case. 
    In particular, $s$ and $t$ belong to the same biconnected component of $G' = G + st$, which by~\Cref{obs:Gst-equals-block-B} is exactly $G_{s,t} + st$ (note that possibly $G_{s,t} = \{st \}$).
    Let us note that $st$ can only be added to $H'$ if it was already present in $G$ (otherwise, if $st \notin E(G)$, then $st$ added temporarily to \tsc{bc} as an~unmarked edge and cannot be inserted into \tsc{tw} by \tsc{LongPath}).
    Hence $H' \subseteq G$.
    Assume first that the \textsf{while} loop is broken when $e = \textsf{Null}$.
    Then, all edges of $(G \setminus H) \cap G_{s,t}$ have been added to $H'$, meaning $H'[V(G_{s,t})] = G_{s,t}$.
    That is, all paths between $s$ and $t$ in $G$ are contained in the graph $H'$ held by \tsc{tw}, and so \tsc{tw.LongestPath}$(s,t)$ correctly outputs the length of the longest $(s,t)$-path in $G$.
    Otherwise, the \textsf{while} loop returns \textbf{Yes} when \tsc{tw} rejects adding some edge of $(G \setminus H) \cap G_{s,t}$, say $uv$, to $H'$.
    That is, we have $\tw(H') \leq 2k$ and $\tw(H' + uv) > 2k$ by~\Cref{lem:dynamic-treewidth}.
    Consider then the graph $H'' = H + G_{s,t} + st$, note that $H' \subseteq H'' \subseteq G'$, and therefore $V(G_{s,t})$ induces the biconnected component $G_{s,t}+st$ in $H''$ (because it did in $G'$).
    Since we had $\tw(H') \leq 2k$, all biconnected components of $H''$ other than $G_{s,t} + st$ still have treewidth at most $2k$, but $\tw(H'') \geq \tw(H' + uv) > 2k$.
    Therefore, since the treewidth of $H''$ is the maximum treewidth of its biconnected components, we must have $\tw(G_{s,t}+st) > 2k$.
    Then, $\tw(G_{s,t}) > 2k-1$ and~\Cref{lem:extract-k-path-high-tw} yields a $k$-path between $s$ and $t$.
    %}
\end{proof}

\thmpath*
\begin{proof}
The fully dynamic graph $G$ our data structure maintains is exactly the $G$ maintained by the biconnectivity data structure, so clearly \tsc{Insert} and \tsc{Delete} are sound. Then \Cref{lem:k-path-query-soundness} shows the soundness of \tsc{LongPath}.
It remains to justify the amortized runtime of all three operations.
Updates \tsc{Insert}, \tsc{Delete} clearly run in amortized time $2^{O(k^3)} \log n + O(\log^2 n \log^2 \log n)$ by~\Cref{thm:biconnectivity-ds,lem:dynamic-treewidth}; %\micr{now the dependence on $k$ is higher for insertion \ms{fixed}}
and \tsc{Insert} adds at most one marked edge in \tsc{bc}.
Moreover, \tsc{LongPath} runs in constant time, plus a~constant number of calls to \tsc{tw} and \tsc{bc}, plus an~additional constant number of calls to \tsc{tw} and \tsc{bc} for each edge unmarked in the \textsf{while} loop; moreover, no new edges become marked in \tsc{LongPath}.
Hence the running time of \tsc{LongPath} is amortized by the decrease in the number of marked edges in \tsc{bc}, and so by~\Cref{thm:biconnectivity-ds,lem:dynamic-treewidth}, all operations run in amortized time $2^{O(k^3)} \log n + O(\log^2 n \log^2 \log n)$.
\end{proof}

\subsection{Dynamic $(s,t)$\tsc{-Detour}}
\label{ssec:dynamic-detour}

\paragraph{Dynamic data structure.}
For any fixed $k$, our data structure for $(\geq k)$-detours instantiates the following auxiliary data structures for a fully dynamic graph $G$:
\begin{itemize}[nosep]
    \item A biconnectivity data structure \tsc{bc} for $G$ given by~\Cref{thm:biconnectivity-ds},
    \item A dynamic treewidth data structure of a \emph{subgraph} $H \subseteq G$ such that $\tw(H) \leq 32k+47$ guaranteed by~\Cref{lem:dynamic-treewidth}.
    %\ms{@Amadeus: commented out the width-of-the-decomposition part, I think we never discuss this guarantee before} %, holding a tree-decomposition of $H$ of width at most \blue{$9(32k+46) + 8$}.
\end{itemize}
Again, for any biconnected component $B$ of $G$, the marked edges of $B$ are exactly $E(G) \setminus E(H)$.

Queries \tsc{Insert} and \tsc{Delete} are identical to the previous~\Cref{alg:insert,alg:delete}.
Then, \tsc{LongDetour} (\Cref{alg:cap-detour}) is almost identical to~\Cref{alg:cap}, given that \tsc{tw} is initialized with a higher treewidth bound. The only difference is essentially the condition on the length of the longest $(s,t)$-path tested in the low-treewidth case.
%\ms{changed slightly} %when we query \tsc{tw} for shortest and longest paths, we explicit it for completeness.
\begin{algorithm}
\caption{LongDetour}\label{alg:cap-detour}
\begin{algorithmic}
\Function{\tsc{LongDetour}}{$s,t$}
\State \HelperEdgeVar $\gets$ \textsf{False}
\If{$st \notin E(G)$}
    \State \tsc{bc.Insert}$(s,t)$
    \State \HelperEdgeVar $\gets$ \textsf{True}
\EndIf
%\State \ms{I prefer $\textsf{removeEdge}$ to be initially false, and only do stuff if $st \notin E(G)$ in the first place?}
%\If{\textbf{not} \tsc{bc.IsBiconnected}$(s,t)$}
%    \IIf{\HelperEdgeVar} \tsc{bc.Delete}$(st)$
%   \State \Return \textbf{No}
   %\State \ms{Remember to remove the edge if needed}
%\EndIf
\While{\textsf{True}}
    \State $e \gets$ \tsc{bc.FindMarkedEdge}$(st)$
    %\State \ms{Something happens if no marked edge exists?}
    \IIf{$e = \sf{Null}$} \textbf{break}
    \If{\tsc{tw.Insert}$(e)$ aborts} \Comment{$\tw(H + e) > 32k+47$}
        %\State \tsc{tw.Delete}$(H,e)$
        \IIf{\HelperEdgeVar} \tsc{bc.Delete}$(st)$
        \State \Return \textbf{Yes}
    \Else
        \State \tsc{bc.Unmark}($e$)
    \EndIf
\EndWhile
\IIf{\HelperEdgeVar} \tsc{bc.Delete}$(st)$
\State \Return $\tsc{tw.LongestPath}(s,t) - \tsc{tw.ShortestPath}(s,t) \geq k$
\EndFunction
\end{algorithmic}
\end{algorithm}

\begin{lemma}\label{lem:k-detour-query-soundness}
    For any $s,t \in V(G)$, \tsc{LongDetour}$(s,t)$ correctly outputs whether $s$ and $t$ are joined by a~path of length at least $\dist_G(s,t)+ k$.
\end{lemma}
\begin{proof}
    The proof follows that of~\Cref{lem:k-path-query-soundness} verbatim except for the following.
    Let $G',H'$ be the graphs held by \tsc{bc} and \tsc{tw} at the end of the run of \tsc{LongDetour}, before (potentially) removing $st$ from \tsc{bc}.
    %The same reasoning ensures that if $s,t$ are not connected, we always output \textbf{No}.
    If the \textsf{while} loop is broken, $\tsc{tw.LongestPath}(s,t) - \tsc{tw.ShortestPath}(s,t)$ correctly outputs the order of the longest $(s,t)$-detour in $G$.
    Otherwise, the \textsf{while} loop returns \textbf{Yes} when we now have $\tw(H' + uv) > 32k+47$ for some edge $uv \in E(G_{s,t})$, by~\Cref{lem:dynamic-treewidth}.
    We argue in the same way that if $\tw(G_{s,t}) > 32k+46$, then~\Cref{thm:extract-k-detour-high-tw} yields a $k$-detour between $s$ and $t$ in $G$.
\end{proof}

The time complexity analysis is exactly the same as in the final proof of~\Cref{thm:k-path-dynamic-DS}, which allows us to conclude the proof of \Cref{thm:detour-dynamic-DS}. % achieves to show the result.
% \thmDetour*

\section{Dynamic Even/Odd Path}
We move on to the description of the data structure allowing querying for the $(s,t)$-paths of a~given parity.
The implementation of this data structure will follow the blueprint laid out in \Cref{sec:dynamic-long-paths}, only that we will use the dynamic bipartiteness data structure in place of dynamic treewidth.
The reason is the well-known observation correlating the parities of $(s,t)$-paths in $G$ to the bipartiteness of the $(s,t)$-relevant part of $G$:

\begin{lemma}[{\cite[Lemma 3]{LapaughP84}}]
    \label{lem:both-parities}
    $G$ contains $(s,t)$-paths of both parities if and only if $G_{s,t}$ is not bipartite.
\end{lemma}
% \ms{Technically, their Lemma~3 only states that ``\textbf{biconnected} $G$ contains $(s,t)$-paths of both parities iff $G$ not bipartite'', but it's obviously the same. Should we care?}

\paragraph{Dynamic data structure.}
We instantiate the following auxiliary data structures:
\begin{itemize}[nosep]
    \item A biconnectivity data structure \tsc{bc} for $G$ given by~\Cref{thm:biconnectivity-ds},
    \item A dynamic bipartiteness data structure \tsc{bp} of a bipartite subgraph $H \subseteq G$ given by~\Cref{lem:dynamic-bipartite}.
\end{itemize}

Analogously to the long path and long detour data structures, inserting (resp., removing) an~edge from the data structure resolves to inserting (resp., removing) the edge --- if possible --- within both \tsc{bc} and \tsc{bp} (\Cref{alg:parity-insert,alg:parity-delete}); and if the insertion of the edge to \tsc{bp} fails, the edge is marked in \tsc{bc}.
We only implement the even path query here; the odd path query is completely analogous.
In order to test whether there exists an~even $(s,t)$-path (\Cref{alg:parity-even-path}), we similarly temporarily insert an~unmarked helper edge $st$ to \tsc{bc}; and iterate the marked edges in the biconnected component containing $st$, progressively adding them to \tsc{bp}.
Again, if all such edges are successfully inserted, then $G_{s,t}$ is bipartite and the parity of the length of the $(s,t)$-path can be verified directly in \tsc{bp}; otherwise, an~even $(s,t)$-path exists by \Cref{lem:both-parities}.

\begin{algorithm}
    \caption{Insert for Even/Odd Path}\label{alg:parity-insert}
    \begin{algorithmic}
        \Function{\tsc{Insert}}{$u,v$}
        \State \tsc{bc.Insert}($u,v$)
        \If{\tsc{bp.Insert}($u,v$) aborts}
            \State \tsc{bc.Mark}($uv$)
        \EndIf
        \EndFunction
    \end{algorithmic}
\end{algorithm}

\begin{algorithm}
    \caption{Delete for Even/Odd Path}\label{alg:parity-delete}
    \begin{algorithmic}
        \Function{\tsc{Delete}}{$uv$}
        \State \tsc{bc.Delete}($uv$)
        \State \tsc{bp.Delete}($uv$)
        \EndFunction
    \end{algorithmic}
\end{algorithm}

\begin{algorithm}
\caption{EvenPath}\label{alg:parity-even-path}
\begin{algorithmic}
\Function{\tsc{EvenPath}}{$s,t$}
\State \HelperEdgeVar $\gets$ \textsf{False}
\If{$st \notin E(G)$}
    \State \tsc{bc.Insert}$(s,t)$
    \State \HelperEdgeVar $\gets$ \textsf{True}
\EndIf
\While{\textsf{True}}
    \State $e \gets$ \tsc{bc.FindMarkedEdge}$(st)$
    \IIf{$e = \sf{Null}$} \textbf{break}
    \If{\tsc{bp.Insert}$(e)$ aborts} \Comment{$H+e$ is not bipartite}
        \IIf{\HelperEdgeVar} \tsc{bc.Delete}$(st)$
        \State \Return \textbf{Yes}
    \Else
        \State \tsc{bc.Unmark}($e$)
    \EndIf
\EndWhile
\IIf{\HelperEdgeVar} \tsc{bc.Delete}$(st)$
\State \Return $\tsc{bp.EvenPath}(s,t)$
\EndFunction
\end{algorithmic}
\end{algorithm}

Note that the operation \tsc{OddPath} is implemented exactly as \tsc{EvenPath}, only that in the case where $G_{s,t}$ is bipartite, we conclude by returning $\tsc{bp.OddPath}(s,t)$ instead of $\tsc{bp.EvenPath}(s,t)$.

The amortized time complexity analysis is analogous to those in \Cref{sec:dynamic-long-paths}, only that each operation in \tsc{bp} runs in amortized $\Oh(\log^2 n)$ time.
Therefore, each operation of our data structure takes amortized $\Oh(\log^2 n \log^2 \log n)$ time.

\section{Conclusions}
We presented efficient data structures for three dynamic path length problems: long $(s,t)$-path, long $(s,t)$-detour, and even/odd $(s,t)$-path.
We now briefly discuss possible improvements and generalizations.

\begin{itemize}
    \item Our strategy to obtain long paths and long detours works more generally for any problem to which the following win/win approach applies: when querying vertices $s,t$, low treewidth of the graph allows us to use a tree decomposition automaton answering the query efficiently, and the high treewidth of the relevant subgraph $G_{s,t}$ immediately yields a \textbf{Yes} or \textbf{No} answer to the query over $s,t$.
    % \ms{Removed ``For example, $\ldots$''}
    % For example, we believe it should be possible to strengthen~\Cref{thm:extract-k-detour-high-tw} (by increasing the treewidth bound) to yield not only a detour of length $k$, but also $k$ paths of different lengths.
    % Then, by extending the dynamic treewidth data structure of \Cref{lem:dynamic-treewidth} to report not only the shortest and the longest $(u,v)$-path, but also some $k$ distinct lengths of $(u,v)$-paths (or all of them if there are fewer than $k$), we can derive a~counterpart of \Cref{thm:detour-dynamic-DS} reporting the existence of $(u,v)$-paths of at least $k$ distinct lengths. \mwr{I would personally drop this paragraph, but the content reads ok}

    \item The dependence of $\Oh(\log^2 n \log^2 \log n)$ on the number $n$ of vertices only stems from the time complexity of the dynamic biconnectivity data structure; therefore, any improvement to the time complexity of this data structure automatically implies a~better running time for each of the data structures presented in this work.

    \item
    We believe that the $2^{\Oh(k^3)}$ factor in the running time of the data structures for long path and long detour is not optimal, and that it could be improved to $2^{\Oh(k \log k)}$ or even $2^{\Oh(k)}$.
    The first improvement would require us to avoid the use of the Bodlaender--Kloks \emph{tree decomposition automaton} tracking the \emph{exact} value of the treewidth of the graph (\Cref{lem:automaton-bodlaender}).
    This seems plausible as we basically need the dynamic treewidth data structure to accept edge insertions as long as they do not make the treewidth of the corresponding biconnected component grow above the threshold.
    Then, to achieve the $2^{\Oh(k)}$ dependence on $k$, one also has to improve the evaluation time of the tree decomposition automaton for the longest $(u,v)$-path  %\mwr{removed ``simple''}
    (\Cref{lem:automaton-max}); this seems to be attainable via Cut \& Count or Gaussian elimination methods (see~\cite[Section 11.2]{platypus}).
    Any improvement past $2^{\Oh(k)}$ is highly unlikely: for $k = n - 1$, the long path problem is equivalent to the Hamiltonian Path problem (with prescribed endpoints), for which no $2^{o(n)}$-time algorithm exists under Exponential Time Hypothesis~\cite[Theorem 14.6]{platypus}.
\end{itemize}

\bibliographystyle{alpha}
\bibliography{bibliography.bib}

\clearpage

\appendix

\section{Dynamic biconnectivity with marked edges}
\label{sec:biconnectivity-ds-marking}

We will now sketch how to implement the marking extension in the dynamic biconnectivity data structure, i.e., we will show \Cref{thm:biconnectivity-ds}.

In their implementation of the data structure, Holm et al.\ reduce the problem of dynamic biconnectivity to the \emph{restricted dynamic tree cover level data structure}.
We omit most details in the exposition of this data structure below, opting to present only the elements relevant to our proof.

Suppose $G$ is a~fully dynamic graph and $F$ is its (dynamic) spanning forest (so $V(F) = V(G)$ and $E(F) \subseteq E(G)$).
An~\emph{$\ell$-level} data structure maintains for every pair of adjacent edges $e_1, e_2 \in E(F)$ the \emph{cover level} of $(e_1, e_2)$, denoted $c(e_1, e_2)$, which is an~integer from $[-1, \ell]$.
% The cover levels are transitive in the following sense: For any three edges $e_1, e_2, e_3 \in E(F)$, all incident to the same vertex of $F$, we have $c(e_1, e_3) \geq \min\{c(e_1, e_2), c(e_2, e_3)\}$.
For any pair $e, f \in E(F)$ of non-adjacent edges in the same connected component of $G$, this induces the cover level $c(e, f)$ as the minimum cover level of any pair of adjacent edges on the unique simple path in $F$ connecting $e$ and $f$.
If $e$ and $f$ are in distinct connected components of $G$, we declare that $c(e, f) = -1$.
For the purposes of this section, it is enough to know that additional requirements placed on the definition of cover levels guarantee that two edges $e, f \in E(F)$ are biconnected in $G$ if and only if $c(e, f) \geq 0$.

The forest $F$, along with its cover levels, is manipulated and queried through an~extensive list of operations, split into two types, \emph{light} and \emph{heavy}.
Holm et al.\ implement the restricted dynamic tree cover level data structure with the following amortized time guarantees:
\begin{lemma}[{\cite[Lemma 16]{HolmNRS2025Arxiv}, informal}]
    \label{lem:holm-biconnectivity-treelevels-basic}
    Let $\ell \in O(\log n)$.
    There exists an~$\ell$-level restricted dynamic tree cover level data structure that processes each heavy operation in amortized $O(\log^2 n \log^2 \log n)$ time, and each light operation in amortized $O(\log n \log^2 \log n)$ time.
\end{lemma}

Moreover, Holm et al.\ implement several extensions of the data structure of \Cref{lem:holm-biconnectivity-treelevels-basic}; one of them is the \emph{marking extension}, which we now describe.
% Three of the operations implemented by Holm et al.\ are light updates {\sf Mark} and {\sf Unmark} and a~light query {\sf FindFirstReach}, which we now describe.
At any point of time, each vertex of $F$ may be arbitrarily \emph{marked} by the user of the data structure.
The user can mark a~vertex $v$ via a~light operation ${\sf Mark}(v)$, and similarly remove a~mark from $v$ via a~light operation ${\sf Unmark}(v)$; neither operation modifies the cover levels of the pairs of edges of the underlying forest.\footnote{Formally, the \emph{marking extension} of the cover level data structure in~\cite{HolmNRS2025Stoc} allows placing \emph{level-$i$ marks} on vertices for each $i \in \{0, 1, \ldots, \ell\}$. However, in our own data structure we will not utilize the marking extension in its full generality. The notion of a~marked vertex in our data structure directly corresponds to \emph{$0$-marked vertices} in~\cite{HolmNRS2025Stoc}.}
Marked vertices can be located in $F$ using a~light query ${\sf FindFirstReach}(p, q)$ with the following semantics: assume that $p, q \in V(F)$ are in the same connected component of $F$, and let $P$ be the unique simple path in $F$ with endpoints $p$ and $q$.
The query returns a marked vertex $v$ that is \emph{$0$-reachable from $P$}: a~vertex for which there exists an~incident edge $e$ and an~edge $f \in E(P)$ such that the cover level of the pair $e, f$ is at least $0$.
If no such vertex exists, the query returns $\bot$.\footnote{The query also provides an~extensive tie-breaking scheme, which is of no relevance here.}
The query does not modify the marks or the underlying forest $F$.

In \cite{HolmNRS2025Arxiv}, it is shown that the marking extension can be implemented on top of the tree cover level data structure:
\begin{lemma}[{\cite[Lemma 16]{HolmNRS2025Arxiv}}]
    \label{lem:holm-biconnectivity-treelevels}
    Let $\ell \in O(\log n)$.
    The $\ell$-level restricted dynamic tree cover level data structure can be modified to support the marking extension within the same amortized time bounds on light and heavy operations.
\end{lemma}

We are now ready to (informally) state the main reduction proved by Holm et al.:
\begin{lemma}[{\cite[Lemma 2]{HolmNRS2025Stoc}, informal}]
    \label{lem:holm-biconnectivity-reduction}
    The dynamic biconnectivity data structure can be implemented using a~restricted dynamic tree cover level data structure with $\ell \in O(\log n)$ levels and marking extension so that, when initialized with an edgeless $n$-vertex graph, any sequence of $m$ updates can be modeled using a~sequence of $O(m \log n)$ light operations and $O(m)$ heavy operations in the cover level data structure.
    The time complexity required for this reduction is $O(m \log^2 n)$.
\end{lemma}

% They also provide the implementation of the $\ell$-level restricted dynamic tree cover level data structure:
% \begin{lemma}[{\cite[Lemma 16]{HolmNRS2025Arxiv}, informal}]
%     \label{lem:holm-biconnectivity-treelevels}
%     Let $\ell \in O(\log n)$.
%     There exists an~$\ell$-level restricted dynamic tree cover level data structure that processes each heavy operation in amortized $O(\log^2 n \log^2 \log n)$ time, and each light operation in amortized $O(\log n \log^2 \log n)$ time.
% \end{lemma}

Note that \cref{lem:holm-biconnectivity-treelevels-basic,lem:holm-biconnectivity-treelevels,lem:holm-biconnectivity-reduction} together imply \cref{thm:biconnectivity-ds-basic}.

In order to show \cref{thm:biconnectivity-ds}, we will extend the reduction described above by essentially deploying a~\emph{man-in-the-middle} attack.
To do so, we initialize a dynamic biconnectivity data structure, itself spawning an instance of restricted dynamic tree cover level \tsc{dt}, and we create an additional private instance \tsc{dp} maintaining the same restricted dynamic tree cover level as \tsc{dt} (except for marks). Then, intercepting all calls to \tsc{dt} (while still performing them in \tsc{dt}), and suitably forwarding them to \tsc{dp}, allows us to emulate marked edges of $G$ as marked vertices in \tsc{dp}.
The details follow below.

\begin{proof}[Proof of \cref{thm:biconnectivity-ds}]
    Let $G$ be a~dynamic graph, updated by edge insertions and removals, along with edge marking and unmarking.
    Our aim is to construct the data structure \tsc{bc} implementing biconnectivity in $G$ with edge marks as required by the statement of \cref{thm:biconnectivity-ds}.
    
    Let $G^{(1)}$ be the (dynamic) $1$-subdivision of $G$, i.e., the graph with vertex set $V(G) \cup \{t_e \mid e \in E(G)\}$ formed from $G$ by replacing each edge $e = uv$ with a~path $ut_ev$.
    We initialize the following data structures:
    \begin{itemize}
        \item $\tsc{bc}^{(1)}$: the basic dynamic biconnectivity data structure for $G^{(1)}$, provided by \cref{thm:biconnectivity-ds-basic}.
        As promised by \cref{lem:holm-biconnectivity-reduction}, the data structure internally keeps, for some $\ell \in \Oh(\log n)$, a~restricted cover level data structure \tsc{dt} with $\ell$ levels and marking extension, maintaining a~dynamic spanning forest $F^{(1)}$ of $G^{(1)}$ and some set of user marks on the vertices of $F^{(1)}$, updated internally by $\tsc{bc}^{(1)}$.

        \item \tsc{dp}: a~restricted cover level data structure \tsc{dp} with $\ell$ levels and marking extension, which will maintain the same dynamic spanning forest $F^{(1)}$ of $G^{(1)}$, but a~different set of user marks (selected by us).
    \end{itemize}

    We will maintain the following invariant on vertex marks in $\tsc{dp}$: the original vertices $V(G)$ are kept unmarked in \tsc{dp}, and we declare that the vertex $t_e \in V(F^{(1)})$ is marked in \tsc{dp} whenever $e \in E(G)$ is marked in $G$.

    Observe that two vertices $u, v \in V(G)$ are biconnected in $G$ if and only if they are adjacent in $G$ or biconnected in $G^{(1)}$; hence biconnectivity in $G$ can be tracked by $\tsc{bc}^{(1)}$ and a~dynamic set data structure maintaining $E(G)$.
    All updates performed by $\tsc{bc}^{(1)}$ to the instance \tsc{dt}, \emph{except} {\sf \tsc{dt}.Mark} and {\sf \tsc{dt}.Unmark}, are forwarded verbatim by us to \tsc{dp}.
    Note that $\tsc{bc}^{(1)}$ does \emph{not} query \tsc{dp} directly (i.e., it still only queries \tsc{dt} to support the biconnectivity queries). 
    Instead, we will use \tsc{dp} to search for marked edges in $G$.

    Note that each edge insertion or removal in \tsc{bc} (i.e., each edge insertion or removal in $G$) can be translated to a~constant number of vertex and edge updates in $G^{(1)}$, maintained by $\tsc{bc}^{(1)}$.
    It remains to implement the updates and queries related to the edge marks in $G$.
    Naturally, marking (respectively, unmarking) an~edge $e \in E(G)$ with \tsc{bc.Mark} (resp. \tsc{bc.Unmark}) is implemented by calling ${\sf \tsc{dp.}Mark}(t_e)$ (resp. ${\sf \tsc{dp.}Unmark}(t_e)$).
    The query $\tsc{bc.FindMarkedEdge}(e)$ (given $e \in E(G)$, return marked $f \in E(G)$ biconnected with $e$) is implemented as follows.
    If $e$ is already marked in \tsc{bc}, we can return $e$.
    Otherwise, let $u$ be an~arbitrary neighbor of $t_e$ in the spanning forest $F^{(1)}$ of $G^{(1)}$ maintained by \tsc{dp}, and let us consider the result of calling ${\sf \tsc{dp.}FindFirstReach}(t_e, u)$.

    \begin{itemize}
        \item If the call returns a~vertex of $G^{(1)}$, then it is a~vertex $t_f$ of the subdivision distinct from $t_e$ (since only subdivision vertices are marked in $G^{(1)}$ and $e$ is unmarked in $G$).
        Moreover, by the definition of ${\sf FindFirstReach}$ and the definition of cover levels in $G^{(1)}$, there exists an~edge $t_fw \in E(F^{(1)})$ incident to $t_f$ such that the edges $t_eu$ and $t_fw$ are biconnected in $G^{(1)}$.
        Thus there exists a~simple cycle in $G^{(1)}$ containing these two edges.
        This cycle corresponds to a~simple cycle in $G$ containing both $e$ and $f$, and so $f$ can be safely returned as a~correct answer to the query.

        \item Now suppose the call returns $\bot$.
        We claim that there is no marked edge in $G$ biconnected with~$e$.
        Suppose otherwise that such an~edge $f$ exists (so, equivalently, $t_f$ is marked in $G^{(1)}$).
        Let $w$ be a~neighbor of $t_f$ in $F^{(1)}$.
        Since $e$ and $f$ are biconnected in $G$, there exists a~simple cycle in $G$ containing both $e$ and $f$, which translates to a~simple cycle in $G^{(1)}$ containing $t_e u$ and $t_f w$ as edges.
        Hence $c(t_e u, t_f w) \geq 0$ in $F^{(1)}$ and thus $t_f$ is $0$-reachable from $t_e u$.
        Since $t_f$ is marked in \tsc{dp}, this yields a~contradiction.
    \end{itemize}

    By \cref{lem:holm-biconnectivity-treelevels-basic,lem:holm-biconnectivity-reduction,lem:holm-biconnectivity-treelevels}, all updates relayed to \tsc{dp} performed by the original biconnectivity data structure $\tsc{bc}^{(1)}$ take amortized $O(\log^2 n \log^2 \log n)$ time.
    Moreover, each operation related to the edge marks in $G$ (marking or unmarking of an~edge of $G$, or searching for a~marked biconnected edge) translates in constant time to a~constant number of light operations in $\tsc{bc}^{(1)}$; therefore, all such operations are supported in additional amortized $O(\log n \log^2 \log n)$ time by \cref{lem:holm-biconnectivity-treelevels-basic,lem:holm-biconnectivity-treelevels}.
\end{proof}

\section{Dynamic treewidth with shortest and longest paths}\label{sec:tree-automata-path-query}

In this section we prove \cref{lem:dynamic-treewidth}. We assume the reader's familiarity with the terminology of boundaried graphs and tree decomposition automata, introduced in~\cite[Appendix~A]{dynamicTWArxiv}. The main part of the proof is the following construction of an automaton, which essentially boils down to writing a standard dynamic programming algorithm for the {\sc{Longest Path}} problem working on a tree decomposition.

\begin{lemma}\label{lem:automaton-max}
For every $t\in \N$, 
there is a tree decomposition automaton $\Aa^\mx_t$ such that for any boundaried tree decomposition $\Tt=(T,\bag,\edges)$ of width $t$ of a boundaried graph $G$, if $\rho$ is the run of $\Aa^\mx_t$ on $\Tt$ and $x$ is the root of $T$, then from the state $\rho(x)$ one can compute, in time $2^{\Oh(t \log t)}$, the following information: for any $u,v\in \bnd G$, what is the maximum length of a simple $(u,v)$-path in $G$? The evaluation time of $\Aa^\mx_t$ is $2^{\Oh(t \log t)}$. 
\end{lemma}
\begin{proof}
	Recall that in the setting of boundaried graphs, we assume that all the vertices come from a common reservoir of vertices $\Omega$. For a given set $B\subseteq \Omega$ with $|B|\leq t+1$, let $S[B]$ be the set of all pairs $(\delta,M)$ such~that
	\begin{itemize}
		\item $\delta$ is a function from $B$ to $\{0,1,2\}$, and
		\item $M$ is a partition of $\delta^{-1}(\{1\})$ into sets of size $2$.
	\end{itemize}
	Note that for any $B$ as above, we have $|S[B]|\leq 2^{\Oh(t\log t)}$. Then we define the state space $Q$ of $\Aa^\mx_t$ to consist of all pairs $(B,f)$, where $B\subseteq \Omega$ is such that $|B|\leq t+1$ and $f$ is a function from $S[B]$ to $\N \cup \{-\infty\}$. %\ama{I don't know what these objects are, just assuming they might want to output infinity now}
    Note that writing down one state of $Q$ boils down to specifying $B$ and $|S[B]|\leq 2^{\Oh(t\log t)}$ integers.
	
	The idea behind this definition is that we want the following assertion ($\bigstar$) to hold: 
	\begin{quote}
		For any  boundaried tree decomposition $\Tt=(T,\bag,\edges)$ of a boundaried graph $G$, if $q=(B,f)$ is the state assigned by the run of $\Aa^\mx_t$ to the root of $T$, then
	\begin{itemize}
		\item $B=\bnd G$; and
		\item for any $(\delta,M)\in S[B]$, $f(\delta,M)$ is the maximum total length of a family $\Ll$ of vertex-disjoint paths in $G$ such that (i) $M=\{\textrm{endpoints of }P\colon P\in \Ll\}$ and (ii) every vertex $b\in B$ is incident to exactly $\delta(b)$ edges of $\bigcup_{P\in \Ll} E(P)$; or $-\infty$ if no such $\Ll$ exists.
	\end{itemize}
\end{quote}
	Note that if we achieve ($\bigstar$), then for any pair of vertices $u,v\in \bnd G$, the length of the longest $(u,v)$-path is equal to the maximum of $f(\delta,\{u,v\})$, where $\delta$ ranges over all functions from $B$ to $\{0,1,2\}$ such that $\delta^{-1}(\{1\})=\{u,v\}$. We output $-\infty$ when there is no $(u,v)$-path.
	
	It remains to specify the initial mapping and the transition function of $\Aa^\mx_t$. For the initial mapping, we simply define it so that ($\bigstar$) is true for one-bag tree decompositions; and it can be computed in time $2^{\Oh(t\log t)}$ by brute force, for we work with a graph on at most $t+1$ vertices. For the transition function, writing it so that ($\bigstar$) can be proved by a bottom-up induction on the tree decomposition is a routine task; we leave it to the reader. In particular, the evaluation time is $2^{\Oh(t\log t)}$.
\end{proof}

The same reasoning as in the proof of \cref{lem:automaton-max} yields also a tree decomposition automaton $\Aa^\mn_t$ that has the same properties, except it provides access to the lengths of the shortest paths between the vertices of $\bnd G$, instead of the longest. The proof can be repeated verbatim except for replacing the word ``maximum'' with ``minimum'', and ``$-\infty$'' with ``$+\infty$''.

The dynamic treewidth data structures~\cite{dynamicTWKorhonenETAL, Korhonen25} for parameter $t$ maintain a tree decomposition of width $\Oh(t)$ under the assumption that the treewidth of the graph remains at most $t$.
It may happen though that treewidth exceeds $t$ after edge insertion but the data structure does not report failure.
To determine for sure if treewidth is at most $t$, one can maintain an automaton that tests this property at the cost of increasing the dependence on $t$ of the evaluation time of the automaton (see~\cite[\S 2.1]{dynamicTWKorhonenETAL}).

\begin{lemma}[\cite{BodlaenderK91}]
\label{lem:automaton-bodlaender}
For every $k \le \ell\in \N$, 
there is a tree decomposition automaton $\Bb_{k,\ell}$ such that for any boundaried tree decomposition $\Tt=(T,\bag,\edges)$ of width $\ell$ of a  graph $G$, if $\rho$ is the run of $\Bb_{k,\ell}$ on $\Tt$ and $x$ is the root of $T$, then from the state $\rho(x)$ one can determine, in time $2^{\Oh(k\ell^2)}$, whether $\tw(G) \le k$. The evaluation time of $\Bb_{k,\ell}$ is $2^{\Oh(k\ell^2)}$.
\end{lemma}

With the automata $\Aa^\mx_t$, $\Aa^\mn_t$ and $\Bb_{k,\ell}$ in place, 
we may proceed to the proof of \cref{lem:dynamic-treewidth}.

\begin{proof}[Proof of \cref{lem:dynamic-treewidth}]
	%By using the same trick as in TODO \ama{not using it in our DS for now, don't think it's needed}, 
    We may assume that every query for the length of a longest/shortest path in $G$ concerns the same fixed pair of vertices $(u^\star,v^\star)$. To achieve this, we add two fresh vertices $u^\star,v^\star$ to the graph, which normally remain isolated during processing the updates. Upon receiving a query about the length of a longest/shortest $(u,v)$-path in $G$, we add edges $u^\star u$ and $vv^\star$, ask for the pair $(u^\star,v^\star)$ instead, and then remove again the edges $u^\star u$ and $vv^\star$. The answer to the query about $(u,v)$ is then the answer to the query about $(u^\star,v^\star)$ decremented by $2$.
	
	Now, we set up the dynamic treewidth data structure of Korhonen~\cite{Korhonen25}, for the width parameter $t$.
    % \mw{
    It processes each update in amortized time $2^{\Oh(t)}\log n$, times $\tau(\Oh(t))$ for maintaining a run of an automaton $\cal Q$ with evaluation time $\tau(\ell)$ on width-$\ell$ tree decomposition. %\mwr{it seems in \cite{Korhonen25} that the term $\Oh(\tau\cdot t)$ is added to the running time but I guess it should be multiplied \ms{I don't see where it's added? I think it's multiplied everywhere. Also why ``$\cdot t$'' within $\Oh$?} \mw{that was taken verbatim from \cite{Korhonen25}. but Tuukka just told me that his statement was imprecise and now it should be ok}}
    The data structure maintains a rooted tree decomposition $\Tt$ of $G$ of width $\Oh(t)$ together with the run of $\cal Q$ on $\Tt$.
    First, we enrich the data structure with the automaton $\Bb_{k,\ell}$ for $k = t$ and $\ell = \Oh(t)$ being the bound on the width of the maintained tree decomposition.
    This incurs a~factor of $2^{\Oh(t^3)}\log n$ per update.
    Whenever edge insertion is accepted, we check using \Cref{lem:automaton-bodlaender} if it makes the treewidth grow beyond $t$.
    If yes, we revert the insertion.
    % }
    
    Next, we enrich the data structure with the automata $\Aa^\mx_t$ and $\Aa^\mn_t$. Since the evaluation time of $\Aa^\mx_t$ and $\Aa^\mn_t$ is $2^{\Oh(t\log t)}$, this data structure handles edge insertions and deletions (and reports if an edge insertion would make the treewidth larger than $t$) in time $2^{\Oh(t^3)}\log n$ and maintains 
    %a rooted tree decomposition $\Tt$ of $G$ of width $\Oh(t)$ together with
    the runs of $\Aa^\mx_t$ and $\Aa^\mn_t$ on $\Tt$. At the cost of increasing the width by $2$, we may assume that $u^\star$ and $v^\star$ are at all times contained in every bag of $\Tt$, and thus $\Tt$ is a boundaried tree decomposition of $G$ with boundary $\{u^\star,v^\star\}$. Therefore, by the properties of $\Aa^\mx_t$ and $\Aa^\mn_t$ provided by \cref{lem:automaton-max}, from the states associated with the root of $\Tt$ we may compute the maximum and the minimum length of an $(u^\star,v^\star)$-path in $G$.
\end{proof}

\end{document}